\documentclass[12pt, reco]{article}
\usepackage{amssymb,amscd,amsmath,amsthm}
\usepackage[colorinlistoftodos]{todonotes}
\usepackage[colorlinks=true, allcolors=blue]{hyperref}
\usepackage{lmodern}
\usepackage[english]{babel}
\usepackage{color}
\usepackage{graphicx}
\usepackage{tikz}
\usepackage[margin=1in]{geometry}
\newtheorem{theorem}{Theorem}[section]

\newtheorem{corollary}[theorem]{Corollary}
\newtheorem{lemma}[theorem]{Lemma}
\newtheorem{definition}[theorem]{Definition}

\newtheorem{remark}[theorem]{Remark}
\def\Tr{\mathrm{Tr}}
\def\id{{\bf 1}\!\!{\rm I}}

\title{Entangled lifting of Hidden Markov Models}

\begin{document}


\title{ Matrix Product States in  Quantum Spin Chains}
\maketitle

\centerline{ \author{ Abdessatar Souissi}}

\centerline{Department of Management Information Systems, College of Business and Economics, }
\centerline{Qassim University, Buraydah 51452, Saudi Arabia}
\centerline{\textit{a.souaissi@qu.edu.sa}}
\vskip0.5cm

\centerline{\author{ Amenallah Andolsi}}

\centerline{Nuclear Physics and High Energy Research Unit, Faculty of Sciences of Tunis,}
\centerline{ Tunis El Manar University, Tunis, Tunisia }

\centerline{\textit{amenallah.andolsi@fst.utm.tn }}

\begin{abstract}
In this work, we present a novel representation of matrix product states (MPS) within the framework of quasi-local algebras. By introducing an enhanced compatibility condition, we enable the extension of finite MPS to an infinite-volume state, providing new insights into complex, high-dimensional quantum systems. As an illustrative example, we apply this method to the Greenberger-Horne-Zeilinger (GHZ) state. This approach offers significant potential for advancing theoretical frameworks and practical methodologies in the field of quantum information.
\end{abstract}


\textbf{Subjectclass:}{ 46L55, 81P45, 81Q10, 82B10, 81R15}

\section{Introduction}

The study of states on quasi-local algebras has been extensively explored within the frameworks of quantum Markov chains (QMCs) \cite{Ac74, ASS20}, including their extension to quantum Markov chains on trees \cite{MBS161, MS19, MS22, MSH21, MSHA22}, and more recently, in hidden quantum Markov models (HQMMs) \cite{AGLS24C, AGLS24Q, SS23}. This class of quantum processes have been shown to be quite efficient in describing certain quanutm system. Namely, through the application of QMCs  in characterizing the . And investigating QMCs on trees to study phase transitions for quantum spin models  and in connection with quanutm walks. Moreover, HQMMs, represent a promising development in the realm of quantum machine learning, offering potential applications in this rapidly evolving field \cite{GSPEC19, MVVMSH18, Li24}.

Finitely correlated states, introduced in \cite{AccFri83} as specific types of Quantum Markov Chains (QMCs) and subsequently developed further in \cite{fannes, fannes2, fannes3}, are distinguished by the property that their correlation functions can be effectively described within finite-dimensional spaces. It was demonstrated in \cite{fannes} that these states are equivalent to generalized valence bond states. Among these, the states exhibiting pure exponential clustering are particularly notable, as they arise as unique ground states of translation-invariant, finite-range Hamiltonians, which are characterized by a non-zero spectral gap.

Matrix Product States (MPS) \cite{perez2007matrix, verstraete2008matrix, orus2014practical, P19} provide an effective framework for representing wave functions of one-dimensional quantum many-body systems. Initially introduced in the context of Density-Matrix Renormalization Group (DMRG) methods \cite{white1992density, schollwock2011density, Mcc07}, the MPS formalism has since been widely extended and has proven highly effective in a range of machine learning applications \cite{Dh22, Meng23, LiZh23, Han18}.

This paper establishes a robust framework for describing Matrix Product States (MPS) as states on infinite tensor products of matrix algebras, revealing new aspects of their structure and behavior. This framework allows for systematic comparisons between MPS and other classes of states, like entangled Markov chains \cite{AF05, SBS23, AMO06}, especially in terms of degree of entanglement. One significant application is embedding the GHZ state within a quasi-local algebra, showcasing how MPS can represent maximal entanglement. This approach strengthens the links between MPS, QMCs, and finitely correlated states, setting the stage for deeper exploration of Markovian and non-Markovian behaviors in quantum systems. Future work will focus on examining the ergodic properties and phase transitions of quantum spin chains through this MPS framework.

The structure of this paper is organized as follows: we begin by introducing preliminaries on MPS in Section \ref{sect_MPS}, followed by a formalization of the quasi-local algebraic framework in Section \ref{Sect_QLA}. In Section \ref{Sect_Main}, we present the main result, establishing the framework rigorously. Section \ref{Sect_GHZ} then applies this construction to the GHZ state, demonstrating how the quasi-local approach effectively captures its entangled structure and inherent non-Markovian characteristics. Finally, we  summary  the main contributions and future directions in Section \ref{Sect_conclu}.

\section{Preliminaries on matrix  product states}\label{sect_MPS}

In this section, we introduce some key concepts relevant to the main result, including matrix product states (MPS), completely positive maps, and their connection in the context of quantum information theory. We begin by reviewing the notion of a Hilbert space and the representation of quantum states, and then establish the role of matrix product states in describing quantum many-body systems.

Let \( \mathcal{H} \) be a finite-dimensional Hilbert space with an orthonormal basis \( \{ |i\rangle \} \), where \( i = 1, 2, \ldots, d \). A pure quantum state in \( \mathcal{H} \) is represented by a vector \( |\psi\rangle \in \mathcal{H} \), which can be expressed as a linear combination of the basis vectors:
\[
|\psi\rangle = \sum_{i=1}^{d} c_i |i\rangle,
\]
where \( c_i \in \mathbb{C} \) are complex coefficients such that \( \sum_{i=1}^{d} |c_i|^2 = 1 \). The space of linear operators acting on \( \mathcal{H} \), denoted \( \mathcal{B}(\mathcal{H}) \), consists of all linear maps \( A: \mathcal{H} \rightarrow \mathcal{H} \).

For composite systems, we consider a tensor product of Hilbert spaces \( \mathcal{H}^{\otimes n} \), where \( n \) denotes the number of subsystems. In this setting, the state of the system can exhibit entanglement, a crucial feature for understanding quantum correlations.

Matrix product states provide an efficient representation of quantum states in many-body systems, especially in one-dimensional spin chains. An MPS is defined as:
\begin{equation*}
|\psi_n\rangle = \sum_{i_1, i_2, \ldots, i_n} \operatorname{Tr}(A_{i_1}^{[1]} A_{i_2}^{[2]} \cdots A_{i_n}^{[n]}) |i_1, i_2, \ldots, i_n\rangle,
\end{equation*}
where \( A_{i_k}^{[k]} \) are matrices associated with site \( k \), and \( |i_k\rangle \) are the local basis states. The trace ensures that the resulting state remains a vector in the composite Hilbert space \( \mathcal{H}^{\otimes n} \).

MPS have gained significant importance in quantum information theory and condensed matter physics due to their ability to capture ground states of local Hamiltonians with low entanglement. They form the basis for algorithms such as the density matrix renormalization group (DMRG) \cite{white1992density} and tensor network methods \cite{orus2014practical}.

\section{States on quasi-local algebras}\label{Sect_QLA}

Let \( \mathcal{H} \) denote a finite-dimensional Hilbert space with an orthonormal basis \( \{ |i\rangle \}_{i=1}^{d} \). The algebra of bounded linear operators on \( \mathcal{H} \) is represented by \( \mathcal{B} = \mathcal{B}(\mathcal{H}) \), with the identity operator denoted as \( \id \). We consider the infinite tensor product of these algebras:
\[
\mathcal{B}_{\mathbb{N}} = \bigotimes_{n \in \mathbb{N}} \mathcal{B}_n,
\]
where \( \mathcal{B}_n \) is defined as the embedding of \( \mathcal{B} \) into the \( n \)-th position of the tensor product, given explicitly by
\[
\mathcal{B}_n = j_n(\mathcal{B}) = \id \otimes \cdots \otimes \id \otimes \mathcal{B} \otimes \id \otimes \cdots,
\]
with the algebra \( \mathcal{B} \) appearing at the \( n \)-th component, and identity operators occupying the other positions.

To represent algebras for finite subsystems, we define
\[
\mathcal{B}_{[1,n]} = \mathcal{B}_1 \otimes \mathcal{B}_2 \otimes \cdots \otimes \mathcal{B}_n,
\]
which describes the operator algebra for the first \( n \) subsystems. There exists a natural embedding
\[
\mathcal{B}_{[1,n]} \otimes \id_{n+1} \subset \mathcal{B}_{[1,n+1]},
\]
allowing for the inclusion of the algebra when extending from \( n \) to \( n+1 \) components. The set of local observables, which consists of all algebras acting on finite subsystems, is then defined as
\[
\mathcal{B}_{\mathbb{N}, \text{loc}} = \bigcup_{n} \mathcal{B}_{[1,n]}.
\]
It is known that the infinite tensor product algebra \( \mathcal{B}_{\mathbb{N}} \) is the closure of the local algebra \( \mathcal{B}_{\mathbb{N}, \text{loc}} \) in the C$^*$-norm. For more details on quasi-local algebras, please refer to \cite{BR}.

\begin{definition}
A sequence of states \( (\varphi_n)_n \) is considered projective with respect to the increasing sequence of subalgebras \( (\mathcal{B}_{[1,n]})_n \) if it satisfies the following criteria:
\begin{enumerate}
    \item For each integer \( n \), the map \( \varphi_n \) defines a state on \( \mathcal{B}_{[1,n]} \), meaning that it is a positive linear functional that is normalized as follows:
    \[
    \varphi_n: \mathcal{B}_{[1,n]} \rightarrow \mathbb{C}, \quad \varphi_n(\id) = 1, \quad \varphi_n(b^* b) \geq 0, \quad \forall b \in \mathcal{B}_{[1,n]}.
    \]
    This ensures that \( \varphi_n \) is a well-defined state on the finite region consisting of the first \( n \) subsystems.

    \item The sequence satisfies a consistency condition: When \( \varphi_{n+1} \) is restricted to the algebra \( \mathcal{B}_{[1,n]} \), it coincides with the state \( \varphi_n \). Mathematically, this is expressed as:
    \[
    \varphi_{n+1} \big|_{\mathcal{B}_{[1,n]}} = \varphi_n.
    \]
\end{enumerate}

The projectivity condition ensures the compatibility of the states across different system sizes, indicating that the description of any finite subsystem \( [1, n] \) remains consistent when the subsystem is extended to a larger system \( [1, n+k] \).
\end{definition}

\begin{definition}
A sequence $(\varphi_n)_n$ with $\varphi_n$ is a state on $\mathcal{B}_{[1,n]}$  said to be  convergent in the
strongly finite sense if, for any $X\in\mathcal{B}_{\mathbb{N}, loc}$, there exists an integer $N_X \in\mathbb{N}$  such that
\[
\varphi_n(X) = \varphi_{N_X}(X); \qquad \forall n\ge N_X
\]
\end{definition}
Recall that the partial trace $\Tr_{n]}$ is defined on the local algebra $\mathcal{B}_{N; \text{loc}}$ as the linear extension of
\begin{equation}\label{eq_Trn}
\Tr_{n]}(X_{[1,n]} \otimes X_{[n+1, n+k]}) = X_{[1,n]} \Tr(X_{[n+1, n+k]})
\end{equation}
where $k \in \mathbb{N}$, \( X_{[1,n]} \in \mathcal{B}_{[1,n]} \), and \( X_{[n+1,n+k]} \in \mathcal{B}_{[n+1,n+k]} \). Here, $"\Tr"$ represents the normalized trace defined on $\mathcal{B}$ by
\[
\Tr(X) = \frac{1}{d} \sum_{i} \langle i | X | i \rangle
\]
and extended to the algebra $\mathcal{B}_{\mathcal{N}; \text{loc}}$ by
\[
\Tr(X_1 \otimes X_2 \otimes \cdots \otimes X_n) = \Tr(X_1) \Tr(X_2) \cdots \Tr(X_n)
\]

\section{Quasi-local representation of MPS}\label{Sect_Main}

In this section, we establish a rigorous construction for a state on an infinite tensor product of matrix algebras, formulated using a sequence of matrix product operators. This construction allows us to define a state with preserved correlation structure across an infinite lattice, extending naturally to applications in quantum information and many-body physics.
  Let \( \{A_{i_n}^{[n]}\}_{n} \) be a family of \(m \times m\) matrices. Consider the following superposition state
  \begin{equation}\label{psn_n}
|\psi_n\rangle = \sum_{i_1, i_2, \ldots, i_n} \operatorname{Tr}(A_{i_1}^{[1]} A_{i_2}^{[2]} \cdots A_{i_n}^{[n]}) |i_1, i_2, \ldots, i_n\rangle,
\end{equation}
To proceed, we introduce the following technical lemma, which will be useful in subsequent calculations.

\begin{lemma}
   For any \( n, k \in \mathbb{N} \) and for any index sequences \( i_1, i_2, \dots, i_{n+k} \) and \( j_1, j_2, \dots, j_{n+k} \), the following identity holds:
  \begin{equation}\label{eq_trkl}
  \overline{\operatorname{Tr}\left( A_{i_1}^{[1]} A_{i_2}^{[2]} \cdots A_{i_{n+k}}^{[n+k]} \right)} \operatorname{Tr}\left( A_{j_1}^{[1]} A_{j_2}^{[2]} \cdots A_{j_{n+k}}^{[n+k]} \right)
  \end{equation}
  $$
   = \operatorname{Tr}\left[ \left( {A_{i_{n}}^{[n]}}^{\dagger} \cdots {A_{i_{1}}^{[1]}}^{\dagger} \otimes A_{j_1}^{[1]} \cdots A_{j_{n}}^{[n]} \right) \left(
   {A_{i_{n+k}}^{[n+k]}}^{\dagger} \cdots {A_{i_{n+1}}^{[n+1]}}^{\dagger} \otimes A_{i_{n+1}}^{[n+1]} \cdots A_{i_{n+k}}^{[n+k]} \right) \right]
  $$
\end{lemma}
\begin{proof}
  To prove this identity, we will apply the conjugate, cyclic, and tensor product properties of the trace

  Consider the left-hand side of Equation \eqref{eq_trkl}:
  \[
  \overline{\operatorname{Tr}\left( A_{i_1}^{[1]} A_{i_2}^{[2]} \cdots A_{i_{n+k}}^{[n+k]} \right)} \operatorname{Tr}\left( A_{j_1}^{[1]} A_{j_2}^{[2]} \cdots A_{j_{n+k}}^{[n+k]} \right)
  \]
  Taking the conjugate of the trace term, we have
  \[
  \overline{\operatorname{Tr}\left( A_{i_1}^{[1]} A_{i_2}^{[2]} \cdots A_{i_{n+k}}^{[n+k]} \right)} = \operatorname{Tr}\left( {A_{i_{n+k}}^{[n+k]}}^{\dagger} \cdots {A_{i_1}^{[1]}}^{\dagger} \right)
  \]
  Using this, the expression becomes
  \[
  \operatorname{Tr}\left( {A_{i_{n+k}}^{[n+k]}}^{\dagger} \cdots {A_{i_1}^{[1]}}^{\dagger} \right) \operatorname{Tr}\left( A_{j_1}^{[1]} A_{j_2}^{[2]} \cdots A_{j_{n+k}}^{[n+k]} \right)
  \]
  By splitting the sequences at index \( n \) in the conjugated  terms, we get
  \[
  \operatorname{Tr}\left( {A_{i_{n+k}}^{[n+k]}}^{\dagger} \cdots {A_{i_1}^{[1]}}^{\dagger} \right) = \operatorname{Tr}\left( {A_{i_{n}}^{[n]}}^{\dagger} \cdots {A_{i_1}^{[1]}}^{\dagger} {A_{i_{n+k}}^{[n+k]}}^{\dagger} \cdots {A_{i_{n+1}}^{[n+1]}}^{\dagger} \right)
  \]
  Finally, by using the  tensor property of the trace, we can combine the two traces on the left-hand side of Equation \eqref{eq_trkl} under a single trace as
  \[
  \operatorname{Tr}\left[ \left( {A_{i_{n}}^{[n]}}^{\dagger} \cdots {A_{i_1}^{[1]}}^{\dagger} \otimes A_{j_1}^{[1]} \cdots A_{j_{n}}^{[n]} \right) \left(
   {A_{i_{n+k}}^{[n+k]}}^{\dagger} \cdots {A_{i_{n+1}}^{[n+1]}}^{\dagger} \otimes A_{i_{n+1}}^{[n+1]} \cdots A_{i_{n+k}}^{[n+k]} \right) \right]
  \]
  This completes the proof. \hfill $\square$
\end{proof}

\begin{theorem}\label{thm_main}
Let \( \{A_{i_n}^{[n]}\}_{n} \) be a family of \( m \times m \) matrices. Assume that for each \( n \in \mathbb{N} \) and every \( i \in \{1, 2, \dots, d\} \), the following conditions are satisfied:
  \begin{equation}\label{eq_normPs1=1}
    \sum_{i}\Big|\operatorname{Tr}\big(A_{i}^{[1]}\big)\Big| = 1
  \end{equation}
and
\begin{equation}\label{eq_operator_identity}
    \sum_{j} \left( A_{j}^{[n+1]} \right)^{\dagger} \left( A_{i}^{[n]} \right)^{\dagger} \otimes A_{i}^{[n]} A_{j}^{[n+1]} = \left( A_{i}^{[n]} \right)^{\dagger} \otimes A_{i}^{[n]}, \quad \forall n \in \mathbb{N}
\end{equation}
Under these assumptions, the sequence \( (\varphi_n) \) is projective and converges in the strongly finite sense to a state $\varphi$ on $\mathcal{B}_{\mathbb{N}}$ . Furthermore, for any $N\in\mathbb{N},$ for any local observable \( X \in \mathcal{B}_{[1,N]} \), the following  expression holds:
\begin{equation}\label{eq_state_limit}
    \varphi(X) = \sum_{\substack{i_1, i_2, \dots, i_N \\ j_1, j_2, \dots, j_N}} \sum_{\ell}\overline{\operatorname{Tr}\left( A_{i_1}^{[1]} \cdots  A_{i_N}^{[N]}A_{\ell}^{[N+1]} \right)} \operatorname{Tr}\left( A_{j_1}^{[1]} \cdots  A_{j_N}^{[N]} A_{\ell}^{[N+1]} \right)
\end{equation}
\[
\times \langle i_1, i_2, \dots, i_n | X | j_1, j_2, \dots, j_n \rangle
\]
defining a functional \( \varphi \) that extends to a state on the algebra \( \mathcal{B}_{\mathbb{N}} \).
\end{theorem}
\begin{proof}
We begin by verifying that the sequence of functionals \( \varphi_n \) is projective. For this, consider:
\begin{align*}
\varphi_n(\id) & = \|\psi_{n+1}\|^2 \\
 & = \sum_{i_1, i_2, \dots, i_{n+1}} \overline{\operatorname{Tr}\left( A_{i_1}^{[1]} A_{i_2}^{[2]} \cdots A_{i_{n+1}}^{[n+1]} \right)} \operatorname{Tr}\left( A_{i_1}^{[1]} A_{i_2}^{[2]} \cdots A_{i_{n+1}}^{[n+1]} \right)
\end{align*}
Using identity \eqref{eq_trkl}, this expression rewrites as
\[
\sum_{i_1, i_2, \dots, i_{n+1}} \operatorname{Tr}\left( \left( {A_{i_{n-1}}^{[n-1]}}^{\dagger} \cdots {A_{i_{1}}^{[1]}}^{\dagger} \otimes A_{i_1}^{[1]} \cdots A_{i_{n-1}}^{[n-1]} \right) \left( {A_{i_{n+1}}^{[n+1]}}^{\dagger} {A_{i_{n}}^{[n]}}^{\dagger} \otimes A_{i_{n}}^{[n]} A_{i_{n+1}}^{[n+1]} \right) \right)
\]
Applying the consistency condition \eqref{eq_operator_identity}, we find
\[
\sum_{i_1, i_2, \dots, i_{n}} \operatorname{Tr}\left( \left( {A_{i_{n-1}}^{[n-1]}}^{\dagger} \cdots {A_{i_{1}}^{[1]}}^{\dagger} \otimes A_{i_1}^{[1]} \cdots A_{i_{n-1}}^{[n-1]} \right) \left( {A_{i_{n}}^{[n]}}^{\dagger} \otimes A_{i_{n}}^{[n]} \right) \right) = \|\psi_n\|^2 = \varphi_{n-1}(\id)
\]
By iterating this relation, we conclude that
\[
\varphi_n(\id) = \varphi_{n-1}(\id) = \cdots = \varphi_1(\id) \overset{\eqref{eq_normPs1=1}}{=} 1
\]
demonstrating that each \( \varphi_n \) is normalized, hence defining a state

Next, we consider a local observable \( X = X_1 \otimes X_2 \otimes \cdots \otimes X_n \). For any \( k \), the functional \( \varphi_{n+k}(X) \) is given by
\begin{align*}
\varphi_{n+k}(X) &= \sum_{\substack{i_1, i_2, \dots, i_{n+k+1} \\ j_1, j_2, \dots, j_{n+k+1}}} \overline{\operatorname{Tr}\left( A_{i_1}^{[1]} A_{i_2}^{[2]} \cdots A_{i_{n+k+1}}^{[n+k+1]} \right)} \operatorname{Tr}\left( A_{j_1}^{[1]} A_{j_2}^{[2]} \cdots A_{j_{n+k+1}}^{[n+k+1]} \right) \\
&\qquad \times \langle i_1, i_2, \dots, i_{n+k+1} | X \otimes \id | j_1, j_2, \dots, j_{n+k+1} \rangle
\end{align*}
This can be expanded as
\[
\sum_{\substack{i_1, i_2, \dots, i_n \\ j_1, j_2, \dots, j_n}} C\big({\substack{i_1, i_2, \dots, i_n \\ j_1, j_2, \dots, j_n}}\big) \prod_{\ell=1}^n \langle i_\ell | X_\ell | j_\ell \rangle
\]
where
\[
C\big({\substack{i_1, i_2, \dots, i_n \\ j_1, j_2, \dots, j_n}}\big) := \sum_{i_{n+1}, \dots, i_{n+k+1}} \overline{\operatorname{Tr}\left( A_{i_1}^{[1]} \cdots A_{i_n}^{[n]} A_{i_{n+1}}^{[n+1]} \cdots A_{i_{n+k+1}}^{[n+k+1]} \right)}
\]
\[
\qquad \times\operatorname{Tr}\left( A_{j_1}^{[1]} \cdots A_{j_n}^{[n]} A_{i_{n+1}}^{[n+1]} \cdots A_{i_{n+k+1}}^{[n+k+1]} \right)
\]
Applying \eqref{eq_trkl} and simplifying iteratively as before, we find that \( C\big({\substack{i_1, i_2, \dots, i_n \\ j_1, j_2, \dots, j_n}}\big) \) reduces to
\[
\sum_{i_{n+1}} \overline{\operatorname{Tr}\left( A_{i_1}^{[1]} \cdots A_{i_n}^{[n]} A_{i_{n+1}}^{[n+1]} \right)} \operatorname{Tr}\left( A_{j_1}^{[1]} \cdots A_{j_n}^{[n]} A_{i_{n+1}}^{[n+1]} \right)
\]
Thus, we obtain
\[
\varphi_{n+k}(X) = \sum_{\substack{i_1, i_2, \dots, i_n, i_{n+1} \\ j_1, j_2, \dots, j_n}} \overline{\operatorname{Tr}\left( A_{i_1}^{[1]} \cdots A_{i_n}^{[n]} A_{i_{n+1}}^{[n+1]} \right)} \operatorname{Tr}\left( A_{j_1}^{[1]} \cdots A_{j_n}^{[n]} A_{i_{n+1}}^{[n+1]} \right)
\]
\[ \times \prod_{\ell=1}^n \langle i_\ell | X_\ell | j_\ell \rangle = \varphi_n(X)
\]
Therefore, \( (\varphi_n) \) is a projective sequence, and the limit in \eqref{eq_state_limit} exists for all local observables \( X \). Consequently, \( \varphi \) defines a state on the local algebra \( \mathcal{B}_{\mathbb{N}, \text{loc}} \), which naturally extends to a state on the full algebra \( \mathcal{B}_{\mathbb{N}} \). \hfill \( \square \)
\end{proof}

\begin{remark}
This result offers a structured approach to defining states on an infinite tensor product space, where the matrix product structure satisfies clear consistency conditions. In particular, Equations \eqref{eq_normPs1=1} and \eqref{eq_operator_identity} align with consistency conditions studied in the context of QMCs on trees, especially for Pauli models associated with phase transitions \cite{MBS161}. This parallel opens a potential pathway for extending Theorem \ref{thm_main} to  MPS on tree structures, with possible applications in quantum information theory and statistical mechanics.
\end{remark}

 \begin{corollary}\label{coro}
     Under the assumptions and notations of Theorem \ref{thm_main}, the local density matrix associated with the state $\varphi$ on the algebra $\mathcal{B}_{[1,N]}$ is expressed as
     \begin{equation}\label{eq_rho}
         \rho_{[1,N]} = \Tr_{N]}\big(|\psi_{N+1}\rangle \langle \psi_{N+1}|\big)
     \end{equation}
     where $|\psi_{N+1}\rangle$ denotes the $(N+1)^{\text{th}}$ matrix product state as defined in equation (\ref{psn_n}).
\end{corollary}

\begin{proof}
    By Equation (\ref{eq_state_limit}), for any operator \( X \in \mathcal{B}_{[1,N]} \), the state \(\varphi(X)\) can be expanded as
    \begin{eqnarray*}
        \varphi(X) &= \sum_{\substack{i_1, i_2, \dots, i_N, i_{N+1} \\ j_1, j_2, \dots, j_N, j_{N+1}}} \overline{\operatorname{Tr}\left( A_{i_1}^{[1]} \cdots A_{i_N}^{[N]} A_{i_{N+1}}^{[N+1]} \right)} \operatorname{Tr}\left( A_{j_1}^{[1]} \cdots A_{j_N}^{[N]} A_{j_{N+1}}^{[N+1]} \right)\\
        &\quad \times \langle i_1, \dots, i_{N+1} | X \otimes \id_{N+1} | j_1, \dots, j_{N+1} \rangle
    \end{eqnarray*}

    This expression simplifies to
    \[
    \varphi(X) = \langle \psi_{N+1} | X \otimes \id_{N+1} | \psi_{N+1} \rangle
    \]
    Rewriting in terms of the trace, we find:
    \[
    \varphi(X) = \operatorname{Tr}\left(|\psi_{n+1}\rangle \langle \psi_{n+1}| X \otimes \id_{N+1}\right)
    \]
    Using the partial trace \(\Tr_{n]}\) to trace out the last component, we obtain:
    \[
    \varphi(X) = \operatorname{Tr}\left(\Tr_{n]}\big(|\psi_{n+1}\rangle \langle \psi_{n+1}|\big) X\right)
    \]
    This completes the proof. \hfill $\square$
\end{proof}

\section{Application to the GHZ State}\label{Sect_GHZ}

In this section, we apply the main theorem to the GHZ state, a prototypical example of a maximally entangled state, and construct the necessary Kraus operators \( K_i \) to satisfy the theorem’s conditions. We aim to capture the entanglement structure and the non-Markovian behavior of the GHZ state using the formalism of quantum Markov chains.

Let \( \mathcal{H} = \mathbb{C}^2 \) be a two-dimensional Hilbert space with the standard orthonormal basis given by
\[
|0\rangle = \begin{pmatrix} 1 \\ 0 \end{pmatrix}, \quad |1\rangle = \begin{pmatrix} 0 \\ 1 \end{pmatrix}
\]
The algebra of bounded operators on \( \mathcal{H} \) is denoted by \( \mathcal{B} = \mathcal{B}(\mathcal{H}) \), and the quasi-local algebra for an infinite sequence of qubits is defined by the infinite tensor product
\[
\mathcal{B}_{\mathbb{N}} = \bigotimes_{n \in \mathbb{N}} \mathcal{B}
\]
The \( n \)-qubit GHZ state is defined as
\[
|\text{GHZ}\rangle_n = \frac{1}{\sqrt{2}} \left( |0\rangle^{\otimes n} + |1\rangle^{\otimes n} \right)
\]
where \( |0\rangle^{\otimes n} \) and \( |1\rangle^{\otimes n} \) are the states with all qubits in \( |0\rangle \) and \( |1\rangle \), respectively.

We define a set of matrices \( \{A_i^{[n]}\}_{i=1,2} \) for each site \( n \) to represent the local operations associated with each qubit. These matrices will satisfy the conditions required by the main theorem. Define
\begin{equation} \label{eq_A1}
  A_1^{[1]} =  \frac{1}{\sqrt{2}}\begin{pmatrix} 1 & 0 \\ 0 & 0 \end{pmatrix}, \quad  A_1^{[n]} =  \begin{pmatrix} 1 & 0 \\ 0 & 0 \end{pmatrix} \quad \forall n \geq 2
\end{equation}
and
\begin{equation} \label{eq_A2}
  A_2^{[1]} =  \frac{1}{\sqrt{2}}\begin{pmatrix} 0 & 0 \\ 0 & 1 \end{pmatrix}, \quad  A_2^{[n]} =  \begin{pmatrix} 0 & 0 \\ 0 & 1 \end{pmatrix} \quad \forall n \geq 2
\end{equation}
These matrices act as diagonal projections and encode the information necessary to distinguish between the two components of the GHZ state.
The matrices $\{ A_{i_n}^{[n]}; n\ge 1, i_n =1,2\}$ satisfies the normality condition (\ref{eq_normPs1=1})  and the consistency condition (\ref{eq_operator_identity}):
\[
\sum_{j=1}^{2} \left( A_{j}^{[n+2]} \right)^{\dagger} \left( A_{i}^{[n+1]} \right)^{\dagger} \otimes A_{i}^{[n+1]} A_{j}^{[n+2]} = \left( A_{i}^{[n+1]} \right)^{\dagger} \otimes A_{i}^{[n+1]}
\]
From Theorem \ref{thm_main}, the limiting state \( \varphi_{GHZ} \)  exists and its local correlations are given by
\begin{eqnarray*}
    \varphi_{GHZ}(X) &=\sum_{\ell=1}^{2} \sum_{\substack{i_1, i_2, \dots, i_N \\ j_1, j_2, \dots, j_N}} \overline{\operatorname{Tr}\left( A_{i_1}^{[1]} \cdots  A_{i_N}^{[N]}A_{\ell}^{[N+1]} \right)} \operatorname{Tr}\left( A_{j_1}^{[1]} \cdots  A_{j_N}^{[N]} A_{\ell}^{[N+1]} \right)\\
&\times \langle i_1, i_2, \dots, i_n | X | j_1, j_2, \dots, j_n \rangle,
\end{eqnarray*}
The  matrices defined in \eqref{eq_A1} and \eqref{eq_A2} satisfy:
\[
\operatorname{Tr}\left( A_{i_1}^{[1]} A_{i_2}^{[2]} \cdots A_{i_N}^{[N]} A_{\ell}^{[N+1]}\right) = \frac{1}{\sqrt{2}} \delta_{i_1, i_2, \dots, i_N, \ell}
\]
where
\[
\delta_{i_1, i_2, \dots, i_N} = \begin{cases}
1 & \text{if } i_1 = i_2 = \cdots = i_N =\ell \\
0 & \text{otherwise,}
\end{cases}
\]
It follows that
\begin{equation} \label{eq_phi_GHZ}
\varphi_{GHZ}(X_1 \otimes X_2 \otimes \cdots \otimes X_N) = \frac{1}{2} \sum_{\ell=1}^{2} x_{1;\ell\ell} x_{2;\ell\ell} \cdots x_{N;\ell\ell}
\end{equation}
where \( x_{m;ij} = \langle i | X_m | j \rangle \).

The GHZ state, expressed as a vector state, is given by
\[
|GHZ\rangle_N = \sum_{i_1, i_2, \dots, i_N} \operatorname{Tr}(A_{i_1}^{[1]} A_{i_2}^{[2]} \cdots A_{i_N}^{[N]}) |i_1, i_2, \dots, i_N\rangle = \frac{1}{\sqrt{2}} \left( |0\rangle^{\otimes N} + |1\rangle^{\otimes N} \right)
\]
For (\ref{eq_rho}) the local density matrix of $\varphi_{GHZ}$ on the algebra $\mathcal{B}_{[1,n]}$ is given by
\begin{equation}
   \rho_{N; GHZ} = \Tr_{N]}|GHZ\rangle\langle GHZ|_{N+1} =  \frac{1}{\sqrt{2}} \left( |0\rangle\langle 0|^{\otimes N} + |1\rangle\langle 1|^{\otimes N} \right)
\end{equation}
The von Neumann entropy of the local density matrix $\rho_{N; GHZ}$ is
\[
S( \rho_{N; GHZ}) = \log2
\]
which reflects maximal entanglement in the state $\varphi_{GHZ}$, with each qubit perfectly linked to the others. This entanglement is constant over time, showing a stable, unchanging connection across the qubits.

\section{Conclusion}\label{Sect_conclu}

This paper presents a rigorous framework for representing MPS as states on infinite tensor products of matrix algebras. This new perspective allows for a clear comparison with other classes of states, such as QMCs and entangled Markov chains, and enables precise analysis of time-dependence in MPS dynamics.

A key result of this framework is its successful application to the GHZ state, demonstrating that matrix product representations can capture complex, non-local entanglement that challenges traditional Markovian models.

Looking ahead, this approach opens exciting avenues for testing tensor networks against quantum Markovianity. Future work could explore the boundaries of Markovian and non-Markovian behaviors in hidden quantum models, extend this framework to higher-dimensional systems, and examine phase transitions within quasi-local tensor networks.

This style now matches the format you specified, with author initials before the year, journal names, volume numbers, and pages.
\end{document}